\documentclass[12pt]{article}

\newtheorem{theorem}{Theorem}

\newtheorem{proposition}[theorem]{Proposition}

\newcommand{\asp}{{\it asp}\ }
\usepackage{amsmath}
\newenvironment{proof}[1][Proof]{\noindent \textbf{#1.} }{\  \rule{0.5em}{0.5em}}
\usepackage{tikz}
\begin{document}

\title{Discrete factor analysis}
\author{Rolf Larsson}
\maketitle
\begin{abstract}
In this paper, we present a method for factor analysis of discrete data. This is accomplished by fitting a dependent Poisson model with a factor structure. To be able to analyze ordinal data, we also consider a truncated Poisson distribution. We try to find the model with the lowest AIC by employing a forward selection procedure. The probability to find the correct model is investigated in a simulation study. Moreover, we heuristically derive the corresponding asymptotic probabilities. An empirical study is also included.\vskip.2cm\noindent
\end{abstract}

\section{Introduction}

The main idea of classical factor analysis (see e.g. J\"oreskog et al, 2016) is to describe a random vector ${\bf X}=(X_1,...,X_n)'$ as a linear combination of unknown factors $\boldsymbol\xi=(\xi_1,...,\xi_k)'$ plus some independent random error $\boldsymbol\delta=(\delta_1,...,\delta_n)'$ say, where $k<n$. Introducing the $n\times k$ matrix of factor loadings $\boldsymbol\Lambda$, we then have the equation
\begin{equation}
{\bf X}=\boldsymbol\Lambda\boldsymbol\xi+\boldsymbol\delta.\label{FA}
\end{equation}
Restricting the $\xi_j$ to be uncorrelated with zero mean and unit variance, and ucorrelated with $\boldsymbol\delta$, the covariance matrix of ${\bf X}$ is 
\begin{equation}
\boldsymbol\Sigma={\boldsymbol\Lambda}{\boldsymbol\Lambda}'+{\boldsymbol\Psi},\label{Sigma}
\end{equation} 
where ${\boldsymbol\Psi}$ is the covariance matrix of  $\boldsymbol\delta$, usually assumed to be diagonal.

The main focus of (explanatory) factor analysis is to find out about the structure of the loading matrix $\boldsymbol\Lambda$. A common way to deal with this is to let ${\bf X}$ be multivariate normal with zero mean and covariance matrix $\boldsymbol\Sigma$ as in (\ref{Sigma}), and estimate the parameters by maximum likelihood, see J\"oreskog (1967). Observe that the $\boldsymbol\Lambda$ matrix is unique only up to rotation, i.e. ${\boldsymbol\Lambda}{\boldsymbol\Lambda}'={\boldsymbol\Lambda}^*{{\boldsymbol\Lambda}^*}'$ for any ${\boldsymbol\Lambda}^*={\boldsymbol\Lambda}{\bf U}$ where ${\bf U}$ is some orthogonal $k\times k$ matrix.

Despite for the many appealing features of maximum likelihood, searching for the `best'  factor analysis model given data involves some more or less `arbitrary' steps such as choosing the number of factors $k$ and a suitable rotation matrix ${\bf U}$.

In applications, it is common that the data $\bf x$ are observed on an ordinal scale. The continuous variable factor analysis model in (\ref{FA}) can still be applied to this situation, see e.g. J\"oreskog and Moustaki (2001). To this end, the observed data are considered as outcomes from an underlying continuous variable (preferrably normal) that may be described by the factor model (\ref{FA}). Here, a certain (integer) value of the data corresponds to an interval on the continuous scale, defined by threshold parameters. As with all the other parameters, the thresholds may be estimated by maximum likelihood. In terms of numerics, this is a quite formidable task. Hence, alternative procedures have been proposed, for example using polychoric correlations, see Olsson (1979), or likelihood approximations, see Katsikatsou et al (2012).

Factor analysis with discrete data is performed by Zhou et al (2012) and Wedel et al (2003). The former proposes a fully Bayesian method where the parameter vector of the observed discrete variates is modelled with a factor structure similar to the classical J\"oreskog model. The latter approach uses a generalized linear regression model, with a link function that is a function of covariates in a factor form. Like the classical method for continuous data, these two approaches are rather involved numerically and contain issues about factor rotation and determination of the number of factors.

In the present paper, we propose a completely different approach to discrete and ordinal data factor analysis. 
The basis of our approach is the dependent Poisson model, described in e.g. Karlis (2003). In particular, let $U$, $X_1$ and $X_2$ be independent Poisson variates. Then, the variates $Y_1=U+X_1$ and $Y_2=U+X_2$  are also Poisson, but they are now linked through the common factor $U$. Of course, this idea may be extended to arbitrary dimensions, and we could also think of a vector of variables $(Y_1,...,Y_N)$ which may be split up into a number of independent sub systems of the type described. This is then a way to construct a discrete factor model. To deal with ordinal data, we consider truncated distributions. To relax the requirement of independent sub systems, we propose a mixed model approach.

In fact, this factor model idea extends to any (combinations of) discrete distributions, but as a start, we only pursue Poisson in the present paper. Observe that there are many other ways to construct dependent systems of discrete random variables, e.g. via copulas, mixing (compound Poisson) and graphical models. However, none of these seems to produce a system with a factor structure. See further Inouye et al (2017) for a recent overview.

As will be seen in the sequel, the construction of factor models in the way outlined here, as well as maximum likelihood estimation of them, is fairly straightforward. The issue that may be problematic and time consuming is how to choose the `best' possible model among the very many possible suggestions for a given dimension $k$. In this paper, by the `best' model we mean the one with the lowest value of the Akaike information criterion (AIC), see Akaike (1974). We propose to resolve this by employing a forward search algorithm. We will study the probability to find the `correct' model (if there is one) by simulations in dimensions five (where we compare to selection among all possible models) and seven, and we also heuristically calculate the corresponding asymptotic probabilities. 

The rest of the paper is as follows. In section 2 we lay out the model and its estimation via maximum likelihood. The selection algorithm is presented and discussed in section 3. Section 4 contains a simulation study. In section 5, we give an empirical example with ordinal data that previously has been analysed by J\"oreskog et al (2016). Section 6 concludes.

\section{Model and estimation}

\subsection{General}

At first, let us repeat the Karlis bivariate model,
\begin{equation}
\left\{\begin{array}{rcl}
Y_1&=&U+X_1,\\
Y_2&=&U+X_2,\end{array}\right.\label{Karlis}
\end{equation}
where $U$, $X_1$ and $X_2$ are independent random variables that may attain non negative integer values. (At this stage, we do not impose the Poisson assumption.) We say that $U$ is the ``common factor'' that ``loads'' on the variables $Y_1$ and $Y_2$.

It is easy to imagine a setup of a number of possibly dependent variables $Y_1,...,Y_N$ which may be ``linked'' by a set of common factors $U_1,...,U_k$ where $k<N$. If these factors are only allowed to load on one variable each, this gives the general model
\begin{equation}
\left\{\begin{array}{rcl}
Y_1&=&X_1,\\
Y_2&=&X_2,\\
&\vdots&\\
Y_{n_0}&=&X_{n_0},\\
Y_{n_0+1}&=&U_1+X_{n_0+1},\\
&\vdots\\
Y_{n_0+n_1}&=&U_1+X_{n_0+n_1},\\
Y_{n_0+n_1+1}&=&U_2+X_{n_0+n_1+1},\\
&\vdots\\
Y_{n_0+n_1+n_2}&=&U_2+X_{n_0+n_1+n_2},\\
&\vdots\\
Y_{n_0+...+n_k+1}&=&U_k+X_{n_0+...+n_{k-1}+1},\\
&\vdots\\
Y_{n_0+...+n_k}&=&U_k+X_{n_0+...+n_k},
\end{array}\right.\label{genmodel}
\end{equation}
where $N=n_0+...+n_k$ and $U_1,...,U_k,X_1,...,X_N$ are all assumed to be independent non negative integer valued random variables. Moreover, we assume that $n_i\geq 2$ for $i=1,2,...,k$. Observe that this setup also allows for a set of independent components $Y_1,...,Y_{n_0}$. In the following, we will refer to this as a model of type $(n_1,n_2,...,n_k,1,...,1)$, where there are $n_0$ ones at the end. The variables may be shuffled around so that $n_1\geq n_2\geq...\geq n_k$. For example, the model of type $(3,2,1,1)$ is given by
\begin{equation}
\left\{\begin{array}{rcl}
Y_1&=&X_1,\\
Y_2&=&X_2,\\
Y_3&=&U_1+X_3,\\
Y_4&=&U_1+X_4,\\
Y_5&=&U_1+X_5,\\
Y_6&=&U_2+X_6,\\
Y_7&=&U_2+X_7.\label{model3211}
\end{array}\right.
\end{equation}

We want to estimate the parameters of (\ref{genmodel}) by maximum likelihood. This is very feasible, since (\ref{genmodel}) consists of $n_0+k$ simultaneously independent systems. Hence, the likelihood is the product of the likelihoods of all these systems, and the maximum likelihood is the product of the corresponding maximum likelihoods, which all may be evaluated separately. For example, in (\ref{model3211}), the likelihood is a product of likelihoods of one three-dimensional system with one common factor, one two-dimensional system with a common factor and two separate one-dimensional variates.

We need to add distributional assumptions on the $U_j$s and $X_j$s. For example (cf Karlis, 2003), we could assume that the $U_j$ are Poisson with parameters $\lambda_j$ and that the $X_j$ are Poisson with parameters $\mu_j$. Then, by the additivity property of the Poisson distribution, the $Y_j$ are also Poisson, but dependent. The degree of dependence, measured by e.g. the correlation, is a function of the parameters. In the simplest example, (\ref{Karlis}) with $U\sim{\rm Po}(\lambda)$ and $X_j\sim{\rm Po}(\mu_j)$ for $j=1,2$, the correlation between $Y_1$ and $Y_2$ is given by $${\rm corr}(Y_1,Y_2)=\frac{\lambda}{\sqrt{(\lambda+\mu_1)(\lambda+\mu_2)}}.$$
Observe that in this way, only positive correlations are allowed for.

In (\ref{Karlis}), if $f(u;\lambda)$ and $g(x;\mu_j)$ are the probability mass functions of $U$ and the $X_j$ respectively, and we have a set of observation pairs $(y_{11},y_{12}),...,(y_{n1},y_{n2})$. Since $Y_1$ and $Y_2$ are conditionally independent given $U$, the likelihood is
\begin{equation}
L(\lambda,\mu_1,\mu_2)=\prod_{i=1}^n \sum_{u=0}^{\min(y_{i1},y_{i2})}f(u;\lambda)g(y_{i1}-u;\mu_1)g(y_{i2}-u;\mu_2).
\label{Karlik}
\end{equation}
Imposing the Poisson assumption, this becomes
\begin{align}
L(\lambda,\mu_1,\mu_2)&=\prod_{i=1}^n \sum_{u=0}^{\min(y_{i1},y_{i2})}\frac{\lambda^ue^{-\lambda}}{u!}\frac{\mu_1^{y_{i1}-u}e^{-\mu_1}}{(y_{i1}-u)!}\frac{\mu_2^{y_{i2}-u}e^{-\mu_2}}{(y_{i2}-u)!}\notag\\
&=e^{-n(\lambda+\mu_1+\mu_2)}\prod_{i=1}^n \sum_{u=0}^{\min(y_{i1},y_{i2})}\frac{\lambda^u}{u!}\frac{\mu_1^{y_{i1}-u}}{(y_{i1}-u)!}\frac{\mu_2^{y_{i2}-u}}{(y_{i2}-u)!}.
\label{Karlikpo}
\end{align}
The right-hand side of (\ref{Karlikpo}) (of rather the log of it) is readily maximized over the parameters with standard numerical iteration methods. In fact, because of proposition 1 below, we only need to maximize over $\lambda$ since it turns out that $\hat\lambda+\hat\mu_k=\bar y_k$ for $k=1,2$ where $\hat\lambda$ and $\hat\mu_k$ are the MLEs of $\lambda$ and the $\mu_k$, respectively. 

For any numerical maximization in this paper, we use the Matlab routine \texttt{fmincon}.

Next, consider a model with one common factor and an arbitrary number of variables, $m$ say, i.e.
\begin{equation}
\left\{\begin{array}{rcl}
Y_1&=&U+X_1,\\
Y_2&=&U+X_2,\\
&\vdots&\\
Y_m&=&U+X_m.\label{modelm}
\end{array}\right.
\end{equation}
Let $f(u;\lambda)$ and $g(x;\mu_j)$ be the probability mass functions of $U$ and the $X_j$ respectively, $j=1,2,...,m$. Then, with $m$-dimensional observations $(y_{i1},...,y_{in})$ for $i=1,2,...,n$, we get the likelihood
\begin{equation}
L(\lambda,\mu_1,...,\mu_m)=\prod_{i=1}^n \sum_{u=0}^{\min(y_{i1},...,y_{im})}f(u;\lambda)g(y_{i1}-u;\mu_1)\cdots g(y_{im}-u;\mu_m),
\label{likgen}
\end{equation}
and, imposing the Poisson assumption,
\begin{align}
&L(\lambda,\mu_1,...,\mu_m)\notag\\
&=e^{-n(\lambda+\mu_1+...+\mu_m)}\prod_{i=1}^n \sum_{u=0}^{\min(y_{i1},...,y_{im})}\frac{\lambda^u}{u!}\frac{\mu_1^{y_{i1}-u}}{(y_{i1}-u)!}\cdots\frac{\mu_m^{y_{im}-u}}{(y_{im}-u)!}.
\label{likgenpo}
\end{align}

Again, to perform numerical maximization of (\ref{likgenpo}) over the parameters, we only need to maximize w.r.t. $\lambda$. This is a simple consequence of the following proposition. (This fact was also pointed out by Karlis, 2003.)

\begin{proposition}
The parameters that maximize (\ref{likgenpo}), $\hat\lambda,\hat\mu_1,...,\hat\mu_m$, satisfy the equalitites
\begin{equation}
\bar y_k=\hat\mu_k+\hat\lambda,\quad k=1,2,...,m,\label{MLeq}
\end{equation}
where $\bar y_k=n^{-1}\sum_{i=1}^n y_{ik}$ for all $k$.
\end{proposition}
\begin{proof}
	See the appendix.
\end{proof}

\subsection{Truncated distributions}

Considering the situation with ordinal data, the Poisson assumption does not seem to fit perfectly well because of the finite number of classes. However, it can still be considered to provide an approximation. Alternatively, the truncated Poisson distribution could be tried. This means that we condition the Poisson variable to at most attain a maximum value, $A$ say. The probability mass function of a ${\rm Po}(\lambda)$ variable truncated in such a way is
$$f(y;\lambda)=\frac{\lambda^y/y!}{\sum_{j=0}^A \lambda^j/j!}.$$
The formulae in the previous section may be readily adjusted to cover this case. However, there does not seem to be any counterpart to proposition 1. Thus, numerical maximization of the likelihood must be performed simultaneously over all parameters, not only over $\lambda$.

Comparing to the traditional factor analysis setup with an underlying  multivariate normal distribution, there are several immediate advantages with our approach: Our model is more explicit and does not take the long route over some underlying continuous distribution. Also, it seems that we may not run into identification and/or factor rotation issues.

The drawback is that we will have to search for the best model within a very large set of possible models. This issue will be discussed at some length in section 3. 

\subsection{A mixed model}

Comparing our setup to traditional factor analysis models, a potential obstacle is the restriction that more than one factor cannot load on the same $Y$ variable. In the literature, an ANOVA like extension to the outlined model here that permits this is proposed, see e.g. Karlis (2003) and Loukas and Kemp (1983). 

For the purposes of the present paper, the ANOVA like model seems to be quite complicated. We suggest another type of model, that extends the model of the previous sections in an easy way and leads to a relatively simple likelihood function. For example, consider the $(3,2)$ model:
\begin{equation}
\left\{\begin{array}{rcl}
Y_1&=&U_1+X_1,\\
Y_2&=&U_1+X_2,\\
Y_3&=&U_1+X_3,\\
Y_4&=&U_2+X_4,\\
Y_5&=&U_2+X_5.\label{model32}
\end{array}\right.
\end{equation}

We can think about this as two groups, the first group $(Y_1,Y_2,Y_3)$ sharing the common factor $U_1$ and the second group $(Y_4,Y_5)$ sharing $U_2$. But maybe $Y_1$ should rather belong to the second group? This would give us the alternative model 
\begin{equation}
\left\{\begin{array}{rcl}
Y_1&=&U_2+X_1,\\
Y_2&=&U_1+X_2,\\
Y_3&=&U_1+X_3,\\
Y_4&=&U_2+X_4,\\
Y_5&=&U_2+X_5.\label{model32b}
\end{array}\right.
\end{equation}
Now, a mixed model that allows for both of these possibilities is a model that is described by (\ref{model32}) with probability $\pi$ and by (\ref{model32b}) with probability $1-\pi$. Such a model may be interpreted as having both factors $U_1$ and $U_2$ loading on $Y_1$. Here, in a sense, $\pi$ describes the extent to which the first factor, $U_1$, is relatively more important than $U_2$ as a loading on $Y_1$. Of course, $\pi=1$ gives us the model (\ref{model32}) as special case, and $\pi=0$ gives us (\ref{model32b}).

As all the other parameters, $\pi$ may be estimated by maximum likelihood. With notation as above, the likelihood for the mixed model described here is
\begin{equation}
L(\pi,\lambda_1,\lambda_2,\mu_1,...,\mu_5)=\prod_{i=1}^n \left\{\pi s_{i1}s_{i2}+(1-\pi)s_{i3}s_{i4}\right\},\label{lik32p}
\end{equation}
where
\begin{align*}
s_{i1}&=\sum_{u_1=0}^{\min(y_{i1},y_{i2},y_{i3})}f(u_1;\lambda_1)g(y_{i1}-u_1;\mu_1)g(y_{i2}-u_1;\mu_2)g(y_{i3}-u_1;\mu_3),\\
s_{i2}&=\sum_{u_2=0}^{\min(y_{i4},y_{i5})}f(u_2;\lambda_2)g(y_{i4}-u_2;\mu_4)g(y_{i5}-u_2;\mu_5),\\
s_{i3}&=\sum_{u_1=0}^{\min(y_{i2},y_{i3})}f(u_1;\lambda_1)g(y_{i2}-u_1;\mu_2)g(y_{i3}-u_1;\mu_3),\\
s_{i4}&=\sum_{u_2=0}^{\min(y_{i1},y_{i4},y_{i5})}f(u_2;\lambda_2)g(y_{i1}-u_2;\mu_1)g(y_{i4}-u_2;\mu_4)g(y_{i5}-u_2;\mu_5).
\end{align*}

\section{Model selection}

\subsection{A proposed method}

When choosing between different models, one may for example use information criteria such as AIC or BIC, see e.g. Akaike (1974) and Schwarz (1978), respectively. 
When possible, sequential likelihood ratio tests may be employed as well. 

In the following, we have chosen to stick to AIC. In presence of data sets of moderately large sizes, this seems to be the most common choice for model selection in the literature. We will use the definition
\begin{equation}
{\rm AIC}=-2\log L_{max}+2p,\label{AIC}
\end{equation}
where $L_{max}$ is the maximum likelihood value and $p$ is the number of parameters. The selected model is the one with lowest AIC.

The main obstacle with our method is that there are so many potential models (combinations of factors). For large dimensions (numbers of $Y$ variables) $N$, it is completely unrealistic to try them all, even for the fastest computer. 

Below, we will consider dimensions 5 and 7. In dimension 5, there are 52 possible models: the (1,1,1,1,1) model, $\binom{5}{2}=10$ models of type (2,1,1,1), $\binom{5}{3}=10$ models of type (3,1,1), $\binom{5}{4}=5$ models of type (4,1), $5*\binom{4}{2}/2=15$ models of type (2,2,1), $\binom{5}{3}=10$ models of type (3,2), and the model of type (5), where the same factor loads on all the five variables.

In dimension 7, it can be shown that the number of possible models is 877. In fact, the number of possible models in dimension $N$ is described by the Bell number, cf Flajolet and Sedgewick (2009), p.560-562. The Bell number gives the number of partitions of the set of integers from 1 to $N$. Calling this number $B_N$, it holds that $\log B_N$ behaves like $N\log N$ as $N$ tends to infinity. Hence, $B_N$ increases with more than an exponential rate with $N$. Thus, for large dimensions, it is not practically feasible to consider all possible models.

The way out of this dilemma is to try some sort of model selection algorithm. In this paper, we suggest to start with the independence model (1,1,...,1), and compare it with all possible (2,1,...,1) models. (A total of $\binom{N}{2}$ models.) If the independence model is the best (has the lowest AIC), the algorithm stops. If not, we go on by estimating all (3,1,...,1) models where the pair of variables that had the same factor in the first step is joined by one of the other variables ($N-2$ models) as well as all (2,2,1,...,1) models where we add a new pair of variables that consists of any two that were not in the first pair ($\binom{N-2}{2}$ models). If none of the (3,1,...,1) of (2,2,1,...,1) models tried is better than the previously chosen (2,1,...,1) model, we stop and choose the previous model. If not, we go on to test new models, and so it goes on. 

The principle in all steps is to take the favorite model of the previous step and then merge any two groups (considering the ones to be groups of their own). For example, if the previously selected model was of type (2,2,1,...,1), the new models tried are of types (3,2,1,...,1), (2,2,2,1,...,1) and (4,1,...,1).

For dimension $N=5$, the algorithm is illustrated in figure 1. Note that in this figure, we have simplified the last step of the algorithm (if it reaches that far) to test model (5) together with (3,2) and (4,1).  

\begin{figure}
\begin{center}
\begin{tikzpicture}
\draw (0,6) node {$(1,1,1,1,1)$};
\draw (0,4) node {$(2,1,1,1)$};
\draw [->,thick] (0,5.5) -- (0,4.5);
\draw (-2,2) node {$(3,1,1)$};
\draw (2,2) node {$(2,2,1)$};
\draw [->,thick] (-0.5,3.5) -- (-1.5,2.5);
\draw [->,thick] (0.5,3.5) -- (1.5,2.5);
\draw (-3,0) node {\small$(5)$};
\draw (-2,0) node {\small$(4,1)$};
\draw (-1,0) node {\small$(3,2)$};
\draw (1,0) node {\small$(5)$};
\draw (2,0) node {\small$(4,1)$};
\draw (3,0) node {\small$(3,2)$};
\draw [->,thick] (-2.2,1.5) -- (-2.8,0.5);
\draw [->,thick] (-2,1.5) -- (-2,0.5);
\draw [->,thick] (-1.8,1.5) -- (-1.2,0.5);
\draw [->,thick] (1.8,1.5) -- (1.2,0.5);
\draw [->,thick] (2,1.5) -- (2,0.5);
\draw [->,thick] (2.2,1.5) -- (2.8,0.5);
\end{tikzpicture}
\caption{Model selection algorithm, dimension 5.}
\end{center}
\end{figure}
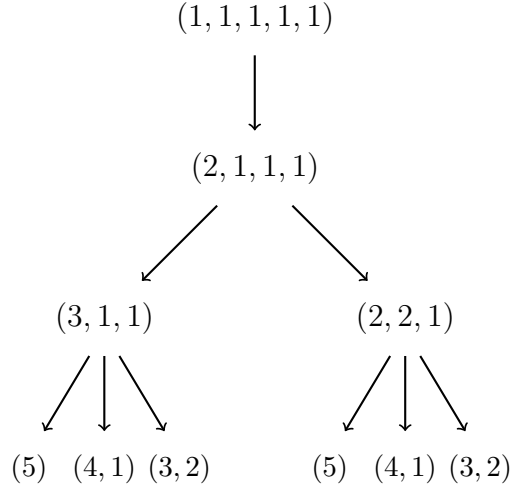

Note that for our algorithm, in dimension five the maximum number of model estimations is 21, out of the 52 possible models. This may not be considered to be a really substantial reduction. However, in dimension seven, the maximum number of model estimations turns out to be 57 out of 877 possible models. 

For an arbitrary dimension $N$, the number of steps in the selection algorithm is of the order $N^3$, see proposition 2. This is in contrast to the exponential rate of increase of the number of possible models as $N$ increases.

\begin{proposition}
In the forward selection algorithm, for dimension $N$ the maximum number of tested models is
\begin{equation}
1+\sum_{k=2}^N\binom{k}{2}=1+\binom{N+1}{3}=\frac{1}{6}(N+2)(N^2-2N+3).\label{numsel}
\end{equation}
\end{proposition}
\begin{proof}
At first, we estimate the model (1,1,...,1). The second step is to estimate all possible (2,1,...,1) models, the number of which is $\binom{N}{2}$. If one of them is the best so far, we go on estimating all models of the forms (3,1,...,1) or (2,2,1,...,1) that may be constructed by merging any two of the $N-1$ subsets in the (2,1,...,1) model. This number of subsets equals $\binom{N-1}{2}$. If each step in the algorithm results in a better model than previously, the procedure goes on until a model with 3 subsets is tested against all of its submodels, the number of which is $4=3+1=\binom{3}{2}+\binom{2}{2}$.

This shows that the maximum number of estimated models is as in the left hand side of (\ref{numsel}). The equalities of (\ref{numsel}) follow from simple algebra.
\end{proof}

\subsection{Asymptotic properties}

In this section, we heuristically derive the asymptotic probabilities to select the correct model for the outlined selection algorithm.

Take dimension 5 as an example. At first, consider testing the model $(1,1,1,1,1)$ vs a specific $(2,1,1,1)$ alternative. Here, the null model has five parameters, while the alternative model has six, the ``extra'' parameter, $\lambda$ say, being that of the common factor. Seeking to minimize AIC (cf (\ref{AIC})), we reject the null model and choose the alternative one if the difference of their $-2$ log likelihood values is more than 2. 

To calculate the asymptotic probability (\asp in the following) for this to happen, we may employ classical results on the maximum likelihood ratio (MLR) test. Here, observe that we are testing $H_0$: $\lambda=0$ vs $H_1$: $\lambda>0$, so we are testing the null that the parameter lies on the boundary of the parameter space. Under $H_0$, the asymptotic distribution of the MLR test is given by e.g. Self and Liang (1987) as $V=Z^2I(Z>0)$ where $Z$ is standard normal and $I$ is the indicator function. In other words, asymptotically and under $H_0$, in our case the \asp {\it not to reject} is given by
\begin{equation}
\gamma=P(V\leq 2)=P(Z<\sqrt{2})\approx 0.92135.\label{Pacc}
\end{equation}

To get further, we need to employ the following unproved postulates (cf Voung, 1989, for a theory of this type for the standard case when the null value of the parameter is not at the boundary of the parameter space):

\begin{enumerate}
\item Tests performed at the same step in the algorithm are asymptotically independent.
\item Tests of a null model with fewer parameters than the alternative model have asymptotic power 1.
\item Tests of a null model with as many as or fewer parameters than the alternative model have asymptotic probability 1 not to reject.
\end{enumerate}

Now, consider testing the $(1,1,1,1,1)$ vs {\it any} $(2,1,1,1)$ model. There are $\binom{5}{2}=10$ alternative models. By postulate 1, we get that the \asp not to reject is $\gamma^{10}\approx 0.44$. 
Hence, we have heuristically derived the \asp to correctly find the $(1,1,1,1,1)$ model to be approximately 0.44.

Next, consider the case when the $(2,1,1,1)$ model is true. Asymptotically, by postulate 2 the probability to get from the $(1,1,1,1,1)$ model to $(2,1,1,1)$ in the first step tends to one. Moreover, this must be the true one among the 10 possible similar models, because from postulate 3, the \asp to accept the true $(2,1,1,1)$ model over a false $(2,1,1,1)$ model is one. 

Coming this far, we will in fact select the true $(2,1,1,1)$ model if we do not reject it when testing vs the three $(2,2,1)$ models that are accomplished by merging the single items as well as testing vs the three $(3,1,1)$ models that we get by putting any single item together with the pair. Testing $(2,1,1,1)$ vs $(2,2,1)$ gives one extra factor parameter, so by postulate 1, the \asp not to reject in any of these three cases is $\gamma^3\approx 0.78$. Testing vs a $(3,1,1)$ model, however, gives no extra parameter, and so, by postulate 3, the \asp to keep the $(2,1,1,1)$ model is one in this case. To sum up, the \asp to correctly select a $(2,1,1,1)$ model is approximately 0.78.

We may go on in the same fashion to calculate the \asp of correctly selecting any possible model. In particular, one may note that the \asp of correctly selecting any model containing at most one single item is one.

All of this was also done for dimension 7.
The \asp values of correct selection are given together with the corresponding finite sample simulated probabilities in the tables of the next section (the $n=\infty$ columns).

\section{Simulations}

The main question to be asked in this section is: What is the probability that the selection algorithm finds the correct model? We check this with simulations. As a start, we will consider dimension 5, where it is feasible to compare the algorithm to the method of estimating all possible models (there are ``only'' 52 of them here). We then go on to dimension 7, which is also the dimension of the empirical example. In this dimension, we only study the selection algorithm. For this dimension, we also check what happens when the distribution is truncated.

All simulations are performed in Matlab2016a. We maximize the likelihood by minimizing the minus log likelihood using the function {\texttt fmincon}. As starting values for the function, we take 0 for the parameters of the factors and the means of the corresponding $Y_i$ observations for the $X_i$.

Inspired by the empirical example in the next section, we take the common factors $U_i$ to be ${\rm Po}(0.5)$. The $X_i$ that sum with a common factor to give the observed $Y_i$ are also ${\rm Po}(0.5)$. For the $X_i$ that are not (so $X_i=Y_i$ in these cases), we take ${\rm Po}(1)$. This means that all $Y_i$ are ${\rm Po}(1)$.

We simulate models of all possible types and then check the proportion of times that AIC is smallest for the model simulated. We also check the proportion of times that the selection algorithm finds the correct model. This always means not only that it is of the correct type, but also that it places the variables correctly into the different groups that have the same common factor.

The results are given in tables 1 and 2. Comparing to testing all models, it is seen that the selection method works remarkably well. As expected, we also find that the selection probabilities increase with $n$, and that they approach the \asp derived in the previous section. Moreover, as is also natural, for models with relatively many factors they are smaller when the parameter is relatively smaller for the factors compared to the independent components. Cf table 2.

\begin{table}
\caption{Estimated probability to find the correct model, dimension 5, parameters 0.5 for the factors and 1 for observed variables, 5000 replicates.}
\vskip.2cm
\begin{tabular}{cccccccc}
&\multicolumn{2}{c}{$n=25$}&\multicolumn{2}{c}{$n=50$}&\multicolumn{2}{c}{$n=100$}&$n=\infty$\\
\hline
model&{\text test all}&{\text selection}&{\text test all}&{\text selection}&{\text test all}&{\text selection}\\
\hline
(5)&0.97&0.97&1.00&1.00&1.00&1.00&1.00\\
(4,1)&0.93&0.91&0.99&0.99&1.00&1.00&1.00\\
(3,2)&0.79&0.78&0.98&0.98&1.00&1.00&1.00\\
(3,1,1)&0.76&0.76&0.90&0.90&0.92&0.92&0.92\\
(2,2,1)&0.62&0.62&0.90&0.90&0.99&0.99&1.00\\
(2,1,1,1)&0.49&0.51&0.66&0.68&0.74&0.76&0.78\\
(1,1,1,1,1)&0.30&0.39&0.32&0.42&0.33&0.42&0.44
\end{tabular}
\end{table}

\begin{table}
\caption{Estimated probability to find the correct model, dimension 5, parameters 0.5 for the factors and 2 for the observed variables, 5000 replicates.}
\vskip.2cm
\begin{tabular}{cccccccc}
&\multicolumn{2}{c}{$n=25$}&\multicolumn{2}{c}{$n=50$}&\multicolumn{2}{c}{$n=100$}&$n=\infty$\\
\hline
model&{\text test all}&{\text selection}&{\text test all}&{\text selection}&{\text test all}&{\text selection}\\
\hline
(5)&0.44&0.45&0.77&0.77&0.96&0.96&1.00\\
(4,1)&0.39&0.36&0.69&0.67&0.92&0.91&1.00\\
(3,2)&0.19&0.18&0.42&0.40&0.76&0.75&1.00\\
(3,1,1)&0.31&0.28&0.54&0.52&0.76&0.76&0.92\\
(2,2,1)&0.12&0.12&0.27&0.27&0.56&0.56&1.00\\
(2,1,1,1)&0.19&0.21&0.31&0.33&0.48&0.50&0.78\\
(1,1,1,1,1)&0.31&0.38&0.33&0.41&0.34&0.42&0.44
\end{tabular}
\end{table}

For dimension 7, we only consider the selection algorithm and one parameter combination, see table 3. The conclusions are similar to dimension 5.

In tables 4 and 5, we consider truncated distributions. The truncation at 3 of table 4 is the same as in the empirical example, whereas the truncation at 2 in table 5 illustrates what happens when the truncation probability is relatively high. To get the same expected value of the $Y$ variables as in the untruncated case, we have chosen parameter values 1.08 and 1.414, respectively, instead of 1 (and half of these values instead of 0.5). As before, the probabilities to find the correct model increase with $n$. Comparing to the case without truncation, we see that the probabilities are smaller, and even more so in case of the more severe truncation in table 5.

\begin{table}
\caption{Estimated probability to find the correct model, the non truncated case, dimension 7, parameters 0.5 for the factors and 1 for the observed variables, 5000 replicates.}
\vskip.2cm
\begin{tabular}{ccccc}
model&$n=25$&$n=50$&$n=100$&$n=\infty$\\
\hline
(7)&0.97&1.00&1.00&1.00\\
(6,1)&0.94&1.00&1.00&1.00\\
(5,2)&0.81&0.97&0.99&1.00\\
(5,1,1)&0.84&0.92&0.91&0.92\\
(4,3)&0.82&0.97&0.99&1.00\\
(4,2,1)&0.70&0.93&0.99&1.00\\
(4,1,1,1)&0.65&0.75&0.77&0.78\\
(3,3,1)&0.73&0.95&0.99&1.00\\
(3,2,2)&0.64&0.96&1.00&1.00\\
(3,2,1,1)&0.54&0.82&0.91&0.92\\
(3,1,1,1,1)&0.43&0.56&0.59&0.61\\
(2,2,2,1)&0.45&0.83&0.98&1.00\\
(2,2,1,1,1)&0.33&0.60&0.74&0.78\\
(2,1,1,1,1,1)&0.22&0.35&0.41&0.44\\
(1,1,1,1,1,1,1)&0.13&0.16&0.16&0.18
\end{tabular}
\end{table}

\begin{table}
\caption{Estimated probability to find the correct model, dimension 7, parameters 0.54 for the factors and 1.08 for the observed variables, truncated at 3, 1000 replicates.}
\vskip.2cm
\begin{tabular}{ccccc}
model&$n=25$&$n=50$&$n=100$&$n=\infty$\\
\hline
(7)&0.93&1.00&1.00&1.00\\
(6,1)&0.87&0.99&1.00&1.00\\
(5,2)&0.69&0.94&0.99&1.00\\
(5,1,1)&0.74&0.90&0.92&0.92\\
(4,3)&0.69&0.93&0.99&1.00\\
(4,2,1)&0.55&0.88&0.97&1.00\\
(4,1,1,1)&0.55&0.73&0.76&0.78\\
(3,3,1)&0.60&0.91&0.98&1.00\\
(3,2,2)&0.48&0.89&1.00&1.00\\
(3,2,1,1)&0.41&0.75&0.90&0.92\\
(3,1,1,1,1)&0.34&0.54&0.56&0.61\\
(2,2,2,1)&0.31&0.73&0.96&1.00\\
(2,2,1,1,1)&0.24&0.51&0.70&0.78\\
(2,1,1,1,1,1)&0.18&0.32&0.37&0.44\\
(1,1,1,1,1,1,1)&0.14&0.14&0.16&0.18
\end{tabular}
\end{table}

\begin{table}
\caption{Estimated probability to find the correct model, dimension 7, parameters 0.707 for the factors and 1.414 for the observed variables, truncated at 2, 1000 replicates.}
\vskip.2cm
\begin{tabular}{ccccc}
model&$n=25$&$n=50$&$n=100$&$n=\infty$\\
\hline
(7)&0.50&0.87&0.98&1.00\\
(6,1)&0.37&0.79&0.97&1.00\\
(5,2)&0.27&0.66&0.94&1.00\\
(5,1,1)&0.35&0.67&0.88&0.92\\
(4,3)&0.29&0.66&0.94&1.00\\
(4,2,1)&0.20&0.54&0.86&1.00\\
(4,1,1,1)&0.23&0.53&0.72&0.78\\
(3,3,1)&0.23&0.57&0.90&1.00\\
(3,2,2)&0.13&0.47&0.91&1.00\\
(3,2,1,1)&0.17&0.44&0.74&0.92\\
(3,1,1,1,1)&0.18&0.37&0.50&0.61\\
(2,2,2,1)&0.09&0.33&0.75&1.00\\
(2,2,1,1,1)&0.09&0.27&0.52&0.78\\
(2,1,1,1,1,1)&0.11&0.20&0.31&0.44\\
(1,1,1,1,1,1,1)&0.11&0.16&0.17&0.18
\end{tabular}
\end{table}

\section{Empirical example}

In this section, we analyze a seven-dimensional data set taken from J\"oreskog et al (2016). The data come from the Eurobarometer Survey of 1992, where citizens of Great Britain were asked about their attitudes towards Science and Technology. The answers are collected on an ordinal scale with values 1,2,3,4. The sample size is $n=392$. The variables are called \texttt{Comfort, Environment, Work, Future, Technology, Industry} and \texttt{Benefit}, but in the following, we will just refer to them as $\tilde y_1,...,\tilde y_7$. Because the means of all variables are closer to 4 than to 1, we have chosen to transform them according to $y_j=4-\tilde y_j$, to get a better fit to a truncated Poisson distribution. The truncation point is then at 3.

In table 6, we give descriptive statistics: mean, variance and the correlation matrix. Observe that the means are larger than the variances. This is in accord with the truncated Poisson distribution. For example, a Poisson variable with parameter 1.08 has expectation 1.00 and variance 0.84 and a parameter value of 0.836 corresponds to expectation 0.80 and variance 0.56. In view of this, we find that most of the variables have a little smaller variances than expected from the truncated Poisson, but not much smaller.

Looking at correlations, it can be seen that some are negative. This is impossible under the dependent Poisson model. However, all negative correlations are small in absolute value. Hence, in a factor analysis context they should be relatively unimportant anyway.

\begin{table}
\caption{Descriptive statistics for the empirical data set (four minus the original data).}
\vskip.2cm
\begin{tabular}{cccrrrrrrr}
&&&\multicolumn{7}{c}{correlations}\\
&mean&variance&$y_1$&$y_2$&$y_3$&$y_4$&$y_5$&$y_6$&$y_7$\\
\hline
$y_1$&0.88&0.35&1.00\\
$y_2$&1.05&0.85&0.08&1.00\\
$y_3$&1.28&0.65&0.15&-0.07&1.00\\
$y_4$&1.01&0.57&0.28&-0.03&0.40&1.00\\
$y_5$&1.00&0.74&0.07&0.39&-0.09&-0.03&1.00\\
$y_6$&0.76&0.58&0.13&0.33&-0.02&0.06&0.35&1.00\\
$y_7$&1.16&0.64&0.33&-0.03&0.17&0.31&-0.01&0.09&1.00\\
\end{tabular}
\end{table}

Next, we try our model selection method, applied for Poisson variables truncated at 3, to the data $(y_1,...,y_7)$. The model found has the same factor structure as the one given in J\"oreskog et al (2016) when estimated with maximum likelihood. It is a $(4,3)$ model where the variables are grouped as $(y_1,y_3,y_4,y_7)$ and $(y_2,y_5,y_6)$. The estimates are found in the first column of table 7. (The standard errors are obtained from the empirical Fisher information, which in turn is calculated as numerical second derivatives of the observed minus log likelihood w.r.t. the parameters. The standard errors are the Fisher informations to the power of $-1/2$.)

We find that the estimates reflect the means of table 6 fairly well. (Recall that the expected value of a truncated Poisson is greater than the parameter.) 


\begin{table}
\caption{Estimated parameters and log likelihood, empirical data set (standard errors in parenthesis).}
\vskip.2cm
\begin{tabular}{ccc}
&model (4,3)&mixed model\\
\hline
$\hat\pi$&-&0.74\ (0.06)\\
$\hat\lambda_1$&0.67\ (0.06)&0.75\ (0.06)\\
$\hat\lambda_2$&0.48\ (0.04)&0.48\ (0.04)\\
$\hat\mu_1$&0.40\ (0.04)&0.37\ (0.04)\\
$\hat\mu_3$&0.89\ (0.06)&0.81\ (0.06)\\
$\hat\mu_4$&0.55\ (0.04)&0.48\ (0.04)\\
$\hat\mu_7$&0.74\ (0.05)&0.66\ (0.05)\\
$\hat\mu_2$&0.70\ (0.05)&0.70\ (0.05)\\
$\hat\mu_5$&0.64\ (0.05)&0.64\ (0.05)\\
$\hat\mu_6$&0.36\ (0.03)&0.36\ (0.03)\\
$\log L$&$-3064.6$&$-3052.3$
\end{tabular}
\end{table}

In J\"oreskog et al (2016) a second model is fitted (using polychoric correlations and weighted least squares).  In this model, $y_1$ is allowed to belong to both variable groups. To see if we can obtain something similar, we fit a mixed model to the data, where $y_1$ belongs to the first group with probability $\pi$. We give the corresponding estimates in the second column of table 7. 

We find that for the mixed model, the log likelihood is more than 2 units higher than the log likelihood for the (4,3) model. Hence, AIC is lower for the mixed model. The interpretation of $\hat\pi=0.74$ is that $y_1$ is more strongly connected to the $(y_3,y_4,y_7)$ group than to $(y_2,y_5,y_6)$. The latter finding is in accord with the estimates of J\"oreskog et al (2016), where $y_1$ loads two to three times stronger on the first group than on the second.

Moreover, observe that $\hat\lambda_1+\hat\mu_j$ for $j=3,4,7$ is about the same for both models and, in fact, they are equal up to two decimal points for $\hat\lambda_2+\hat\mu_j$ for $j=2,5,6$. 


\section{Concluding remarks}

In this paper, we have proposed a method for performing factor analysis on discrete data. In principle, the method should work for any choice of discrete distribution. As a first try, we have chosen the Poisson distribution. Among the very many candidate models, we look for the one with smallest AIC in a forward search algorithm. We have found, both by heuristic calculations and simulations, that this method works well in the sense that it has a high probability to find the correct model (if there is one) for moderately large to large sample sizes.

Since most real life examples of discrete data factor analysis concern ordinal data, we modify our method to deal with this by looking at truncated discrete distributions. So far, ordinal data factor analysis has been performed as in J\"oreskog et al (2016),
who assume an underlying normal distribution. Numerically, the J\"oreskog methodology can be rather complicated, at least when employing maximum likelihood, because in addition to the parameters of main interest, the threshold parameters need to be estimated. Also, as is always the case with traditional factor analysis, factor rotations of more or less arbitrary nature are imposed. 

The method proposed in the present paper is more straightforward. The model is fully specified once the factor structure has been found. The difficulty lies in finding this structure among the very many possible ones. To this end, we have outlined a forward selection method which seems to work well for small and moderately large dimensions. However, model selection for very large dimensions seems to be a challenge that calls for further development of the selection method. This issue is left for future research.

Another aspect that needs further investigation is the choice of distribution. One could of course replace the Poisson by something else like the geometric, the binomial or the negative binomial distribution. Different mixture distributions (different distributions on factors and independent components) are also possible, for example the mixture of the Binomial and Poisson distributions as discussed in Karlis (2003) among others. 

Also, other information criteria than AIC could be used, e.g. BIC. Moreover, in many applications it would be helpful to avoid the requirement that correlations can not be negative, see e.g. Famoye (2015) and Berkhout and Plug (2004). On the theoretical side, a full proof that the heuristic calculations on asymptotic probabilities to find the correct model are valid is called for.

\section*{References}

Akaike, H. (1974), A new look at the statistical model identification. \emph{IEEE Transactions on Automatic Control} 19, 716-723.
\vskip.2cm\noindent
Berkhout, B., Plug, E. (2004), A bivariate Poisson count data model using conditional probabilities. \emph{Statistica Neerlandica} 3, 349-364.
\vskip.2cm\noindent
Famoye, F. (2015), A multivariate generalized Poisson regression model. \emph{Communications in Statistics - Theory and Methods} 44, 497-511.
\vskip.2cm\noindent
Flajolet, P., Sedgewick, R. (2009), \emph{Analytic Combinatorics}, Cambridge Univerisity Press: Cambridge.
\vskip.2cm\noindent
Inouye, D.I., Yang, E., Allen, G.I., Ravikumar, P. (2017), A review of multivariate distributions for count data derived from the Poisson distribution. \emph{WIREs Comput. Stat. 2017, 9:e1398}, doi: 10.1002/wics.1398.
\vskip.2cm\noindent
J\"oreskog, K.G. (1967), Some contributions to maximum likelihood factor analysis. \emph{Psychometrica}, 32, 443-482.
\vskip.2cm\noindent
J\"oreskog, K.G., Moustaki, I. (2001), Factor analysis of ordinal variables: A comparison of three approaches. \emph{Multivariate Behavioral Research}, 36, 347-387.
\vskip.2cm\noindent
J\"oreskog, K.G., Olsson, U.H., Wallentin, F.Y. (2016), \emph{Multivariate Analysis with LISREL}, Springer.
\vskip.2cm\noindent
Katsikatsou, M., Moustaki, I., Yang-Wallentin, F., J\"oreskog, K. G. (2012), Pairwise likelihood estimation for factor analysis models with ordinal data. \emph{Computational Statistics and Data Analysis} 56, 4243-4258.
\vskip.2cm\noindent
Karlis, D (2003), An EM algorithm for multivariate Poisson distribution and related models. \emph{Journal of Applied Statistics} 30, 63-77.
\vskip.2cm\noindent
Loukas, S., Kemp, C.D. (1983), On computer sampling from trivariate and multivariate discrete distributions. \emph{Journal of Statistical Computation and Simulation} 17, 113-123.
\vskip.2cm\noindent
Olsson, U. (1979), Maximum likelihood estimation of the polychoric correlation coefficient. \emph{Psychometrica} 44, 443-460.
\vskip.2cm\noindent
Schwarz, G. (1978), Estimating the dimension of a model. \emph{The Annals of Statistics} 6, 461-464.
\vskip.2cm\noindent
Self, S.G., Liang, K.Y. (1987), Asymptotic properties of maximum likelihood estimators and likelihood ratio tests under nonstandard conditions. \emph{Journal of the American Statistical Association} 82, 605-610.
\vskip.2cm\noindent
Voung, Q.H. (1989), Likelihood ratio tests for model selection and non-nested hypotheses. \emph{Econometrica} 57, 307-333.
\vskip.2cm\noindent
Wedel, M., B\"ockenholt, U., Kamakura, W.A. (2003), Factor models for multivariate count data. \emph{Journal of Multivariate Analysis} 87, 356-369.
\vskip.2cm\noindent
Zhou, M., Hannah, L.A., Dunson, D.B., Carin, L. (2012), Beta-negative binomial process and poisson factor analysis. In \emph{Proceedings of the 15th International Conference on Artificial Intelligence and Statistics}.

\section*{Appendix: Proof of proposition 1}

Suppose that it is not the case that all $y_{i1}=\min(y_{i1},...,y_{im})$. 
Rewrite (\ref{likgenpo}) as
\begin{align}
&L(\lambda,\mu_1,...,\mu_m)\notag\\
&=e^{-n(\lambda+\mu_1+...+\mu_m)}\sum_{u_1=0}^{z_1}...\sum_{u_n=0}^{z_n}\frac{\lambda^{\sum_{j=1}^n u_j}}{\prod_{i=1}^n u_j!}g_1(\mu_1)\cdots g_m(\mu_m),\label{Likproof}
\end{align}
where $z_i=\min(y_{i1},...,y_{im})$ for all $i$ and
$$g_k(\mu_k)=\frac{\mu_k^{n\bar y_k-\sum_{i=1}^n u_i}}{\prod_{i=1}^n (y_{ki}-u_i)!},\quad k=1,2,...,m.$$
Without loss of generality, pick $k=1$. Now, suppressing the arguments of $L$, differentiation w.r.t $\mu_1$ in (\ref{Likproof}) yields
\begin{align}
\frac{\partial L}{\partial\mu_1}&=-nL\notag\\&+e^{-n(\lambda+\mu_1+...+\mu_m)}\sum_{u_1=0}^{z_1}...\sum_{u_n=0}^{z_n}\frac{\lambda^{\sum_{j=1}^n u_j}}{\prod_{i=1}^n u_j!}\left\{\frac{\partial}{\partial\mu_1}g_1(\mu_1)\right\}\notag\\&\cdot g_2(\mu_1)\cdots g_m(\mu_m),
\label{Likder}
\end{align}
where 
$$\frac{\partial}{\partial\mu_1}g_1(\mu_1)=\frac{n\bar y_1-\sum_{i=1}^n u_i}{\mu_1}g_1(\mu_1).$$
Hence, inserting into (\ref{Likder}) and in view of (\ref{Likproof}),
\begin{equation}
\frac{\partial L}{\partial\mu_1}=-nL+\frac{n\bar y_1}{\mu_1}L-\frac{1}{\mu_1}e^{-n(\lambda+\mu_1+...+\mu_m)}\sum_{i=1}^n A_i,\label{Likder1}
\end{equation}
where e.g.
\begin{equation}
A_n=\sum_{u_1=0}^{z_1}...\sum_{u_{n-1}=0}^{z_{n-1}}\frac{\lambda^{\sum_{j=1}^{n-1} u_j}}{\prod_{i=1}^{n-1} u_j!}\sum_{u_n=0}^{z_n}u_n\frac{\lambda^{u_n}}{u_n!}g_1(\mu_1)\cdots g_m(\mu_m).\label{An}
\end{equation}
But since
$$u_n\lambda^{u_n}=\lambda\frac{d}{d\lambda}\lambda^{u_n},$$
it follows from (\ref{An}) and analogous equations for the other $A_i$ that
$$
\sum_{i=1}^n A_i=\lambda\frac{d}{d\lambda}\sum_{u_1=0}^{z_1}...\sum_{u_n=0}^{z_n}\frac{\lambda^{\sum_{j=1}^n u_j}}{\prod_{i=1}^n u_j!}g_1(\mu_1)\cdots g_m(\mu_m).
$$
In view of (\ref{Likproof}), this is
$$
\sum_{i=1}^n A_i=\lambda\frac{d}{d\lambda}\left\{e^{n(\lambda+\mu_1+...+\mu_m)}L\right\}
=\lambda\left\{ne^{n(\lambda+\mu_1+...+\mu_m)}L+e^{n(\lambda+\mu_1+...+\mu_m)}\frac{dL}{d\lambda}\right\},
$$
and (\ref{Likder1}) yields
$$\frac{\partial L}{\partial\mu_1}=-nL+\frac{n\bar y_1}{\mu_1}L-\frac{\lambda}{\mu_1}\left(nL+\frac{dL}{d\lambda}\right).$$
But since $dL/d\lambda=0$ at $\lambda=\hat\lambda$, we get for this $\lambda$ that
$$\frac{\partial L}{\partial\mu_1}=nL\left(-1+\frac{\bar y_1}{\mu_1}-\frac{\hat\lambda}{\mu_1}\right)=\frac{nL}{\mu_1}(-\mu_1+\bar y_1-\hat\lambda),$$
which is zero for $\mu_1=\hat\mu_1=\bar y_1-\hat\lambda$, as was to be shown.

The proof for the case that all $y_{i1}=\min(y_{i1},...,y_{im})$ follows similarly.

\end{document}